\documentclass[11pt]{llncs}
\usepackage[margin=1in]{geometry}
\usepackage[english]{babel}
\usepackage[utf8]{inputenc}
\usepackage{amsmath}
\usepackage{color}
\usepackage{graphicx}
\usepackage{amsfonts}
\usepackage{hyperref}

\usepackage{amsthm}
\usepackage[linesnumbered,ruled]{algorithm2e}
\usepackage{algpseudocode}
\usepackage{pifont}
\usepackage{stackengine}
\usepackage[font=small,skip=5pt]{caption}

\newtheorem{coro}{Corollary}

\newtheorem{thm}{Theorem}
\newtheorem{propname}[thm]{Proposition}

\theoremstyle{plain} 
\newcommand{\thistheoremname}{}
\newtheorem{genericthm}[thm]{\thistheoremname}

\newcommand{\reNCA}{NetREX }
\newcommand{\reNCAs}{NetREX}
\def\tcr{\textcolor{black}}
\title{NetREX: Network Rewiring using EXpression - Towards Context Specific Regulatory Networks}

\author{Yijie Wang\inst{1}\fnmsep\thanks{Equal contribution} \and Dong-Yeon Cho\inst{1}\fnmsep$^\star$ \and Hangnoh Lee\inst{2} \and Justin Fear\inst{2} \and Brian Oliver\inst{2} \and Teresa M Przytycka\inst{1}}
\institute{National Center of Biotechnology Information, National Library of Medicine, NIH, Bethesda MD 20894, USA \email{przytyck@ncbi.nlm.nih.gov} \and
Laboratory of Cellular and Developmental Biology, National Institute of Diabetes and Digestive and Kidney Diseases, 50 South Drive, Bethesda, MD 20892, USA
}

\date{}

\begin{document}
\maketitle

\begin{abstract}
Understanding gene regulation is a fundamental step towards understanding of how cells  function and respond to environmental cues and perturbations. An important step in this direction is the ability to infer the transcription factor (TF)-gene regulatory network (GRN). However gene regulatory networks are typically constructed disregarding the fact that  regulatory programs are conditioned on tissue type, developmental stage, sex, and other factors. Due to lack of the biological context specificity, these context-agnostic networks may not provide insight for revealing the precise actions of genes for a specific biological system under concern. Collecting multitude of features required for a reliable construction of GRNs such as physical features (TF binding, chromatin accessibility) and functional features (correlation of expression or chromatin patterns) for every context of interest is costly. Therefore we need methods that is able to utilize the knowledge  about  a context-agnostic  network (or a network constructed in a related context) for  construction of a context specific regulatory network.  \\\\

To address this challenge we developed a computational approach that utilizes expression data obtained in a specific biological context such as a particular development stage, sex, tissue type and a GRN constructed in a different but related context (alternatively an incomplete or a noisy network for the same context) to construct a  context specific GRN. Our method, \reNCAs, is inspired by network component analysis (NCA) that estimates TF activities and their influences on target genes given predetermined topology of a TF-gene network. To predict a network under a different condition, \reNCA removes the restriction that the topology of the TF-gene network is fixed and allows for adding and removing edges to that network. To solve the corresponding optimization problem, which is non-convex and non-smooth, we provide a general mathematical framework allowing use of the recently proposed Proximal Alternative Linearized Maximization technique and prove that our formulation has the properties required for convergence. \\\\

We tested our \reNCA on simulated data and subsequently applied it to gene expression data in adult females from 99 hemizygotic lines of the \textit{Drosophila} deletion (DrosDel) panel. The networks predicted by \reNCA showed higher biological consistency than alternative approaches. In addition, we used the list of recently identified targets of the Doublesex (DSX) transcription factor  to demonstrate the predictive power of our method.

\end{abstract}

\newpage
\setcounter{page}{1}

\section{Introduction}

Cell function, fitness, and survival depend on a complex regulatory program  involving interactions between genes and their regulators. Regulatory relationships  between transcription factors (TFs) and genes they regulate constitute a gene regulatory network (GRN) that is often represented as a directed bipartite graph.  Several experimental and computationally derived types of evidences  can be used to infer the topology of such regulatory network including genome-wide chromatin immunoprecipitation (ChIP), gene expression profiling, and motif analysis. Complementing the topology of a GRN, a further level of understanding can be obtained by  modeling the quantitative relation  between TF activities and expression of their target genes. Specifically, given expression data obtained by  engineered  perturbations of a reference state, or by tracing expression changes over a number of naturally occurring conditions, the goal is to model the expression changes as a function of changes in activities of TFs when the underlying GRN topology is available.   In particular,  network component analysis (NCA) has proven to be a powerful method for such modeling~\cite{Misra2013,Arrieta-Ortiz2015,Buescher2012,Bar-Joseph2012}.  The essence of NCA approach is the assumption that the expression of genes can be modeled by linear combination of TF activities~\cite{Liao2003}. TF activity is a hidden parameter of each TF that the method infers from the data. While the assumption of linearity of the effects is obviously a simplifying one, it provides a good first approximation and makes the problem tractable.

An important drawback of the NCA approach is that it requires  the topology of the GRN to be known. 
Several computational methods which integrate diverse functional genomics data sets, were developed to infer GRNs and investigate gene regulation at the systems level ~\cite{Marbach2012,Novershtern2011}.
Yet,  current knowledge of the topology of regulatory networks is not complete, even for simple unicellular organisms such as yeast~\cite{Hughes2013} and for most organisms construction of a regulatory network has not even been attempted. 
In addition these networks are typically context-agnostic, namely, they were constructed without considering tissue type, development stage, and other relevant conditions. 
However, for very closely related organisms their regulatory networks can be assumed to be rather similar due to evolutionary conservation.  Similarly, for any specific organism, the regulatory interaction in different tissues are expected to  overlap significantly. This motivates the need for developing a method that can utilize a network constructed for a closely related organism, stage, or tissue as a starting point for constructing a tissue, stage or organism specific network.  
Indeed, there is an increasing recognition of the importance of tissue specific analyses and tissue specific networks~\cite{Park2015,Greene2015}.



To address this challenge we developed Network Rewiring using EXpression (\reNCAs), a new mathematically rigorous method, that builds on the linear model  utilized by the NCA method, but without the assumption that the  topology of the regulatory network is  fixed. That is, unlike previous methods, \reNCA does not restrict the structure of the regulatory network to be hardwired but instead  utilizes expression data from a set of perturbations performed in a given context  and a prior network that is assumed to be related to the target network by limited changes in the topology to construct a context specific network. We remark that allowing for rewiring in the topology of the prior network adds a whole new level of complexity. Specifically, we use $\ell_0$ norm to directly handle the number of removed and newly added edges as well as induce sparse solutions in our formulation. Unlike the widely used strategy, which is replacing the non-convex $\ell_0$ norm by its convex relaxation $\ell_1$ norm, we focus on the harder problem involving $\ell_0$ norm and provide a number of rigorous derivations and results allowing us to adopt the recently proposed  Proximal Alternative Linearized Maximization (PALM) algorithm~\cite{Bolte2014}. In addition, we also proved the convergence of the \reNCA algorithm.  


To evaluate method's performance we first tested \reNCA on simulated data. Specifically, we analyzed its performance as the function of the noise in the prior network and  in the gene expression data. We found that under the assumptions of the model, \reNCA is able to dramatically improve the accuracy of the regulatory networks as long as the prior network and the gene expression are not very noisy.  After testing the method on simulated data,  we  used \reNCA for constructing regulatory networks for adult female flies. We used the network constructed in~\cite{Marbach2012} as the prior network. This network was build by integrating diverse data sets including TF binding, evolutionarily conserved sequence motifs, gene expression across developmental time-course, and chromatin modification data sets. The topology of the network provides an initial wiring diagram that includes TF-target gene interactions predicted from data available at the time of network construction disregarding the context. Starting with this network, we utilized a new expression data set that we collected for adult female flies where perturbations in expression were achieved  by genetic deletions. Specifically, the gene expression data of adult females are from 99 hemizygotic lines (deletion/+) of the \textit{Drosophila} deletion collection (DrosDel) project~\cite{Ryder2004,Ryder2007}. To evaluate the resulting networks we used a previously applied method~\cite{Marbach2012} to access biological relevance of the networks by using Gene Ontology (GO) annotations~\cite{Blake2015} and physical protein-protein interactions (PPIs). \tcr{We compared \reNCA with several methods including a correlation based algorithm~\cite{Marbach2012a} and GENIE3~\cite{Huynh-Thu2010}, the best performer in the DREAM4 \textit{In Silico Multifactorial challenge}~\cite{Marbach2012a}.} We observed dramatic improvements in terms of fold enrichment comparing to all competing algorithms. Subsequently we asked how well the method predicts targets of the TF that have specific roles in adult flies and those targets would be difficult to identify based on cell lines or embryonic data that were predominately used in the construction of the prior network. For this analysis we focused on the Doublesex transcription factor (DSX) whose predicted targets have been recently elucidated ~\cite{Clough2014}.
We showed that the target genes predicted by \reNCA are in good agreement with the experimentally identified  targets.  

Our method addresses an important challenge in analysis of gene regulation. It can be applied in many diverse setting including construction of  condition specific GRNs and networks for organisms related to a model organism where a preliminary regulatory network exits.  As a spin-off of these studies we also developed mathematical underpinning allowing to adopt Proximal Alternative Linearized Maximization (PALM) algorithm to the context of the $\ell_0$ elastic net.

\section{Mathematical Underpinning of the \reNCA Method}
Before describing   the mathematical foundations behind the \reNCA method, we provide a brief overview of the traditional (static) NCA method and its various implementations. Next we introduce the formula for the objective function in our \reNCA method. Importantly, the objective function is non-convex and non-smooth because of using $\ell_0$ norm in our formulation. Rather than relaxing the problem by replacing the non-convex $\ell_0$ norm with the convex $\ell_1$ norm, we have directly solved the more challenging problem with $\ell_0$ norm by adopting the recently proposed Proximal Alternative Linearized Maximization (PALM) algorithm~\cite{Bolte2014} to the original formulation of the problem.

\begin{figure}[htpb!]
\includegraphics[width=16.5cm,height=4.8cm]{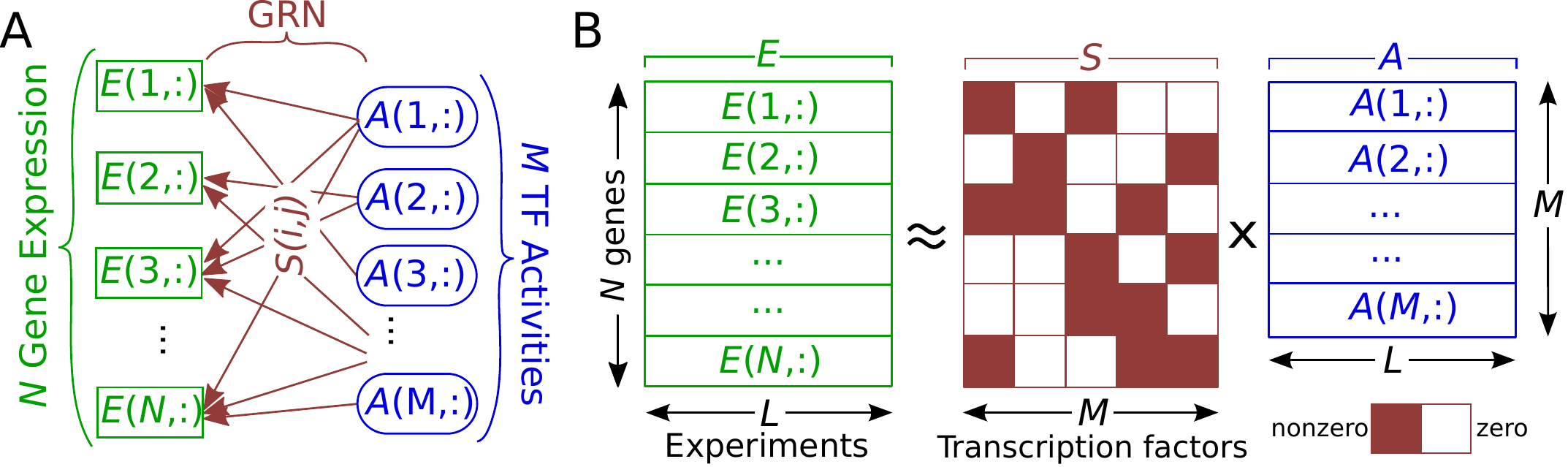}
\caption{\scriptsize The NCA model. (A) The graph representation of NCA. $E(i,:)$ is the expression of gene $i$ over  $L$ experiments and $A(i,:)$ is the activity of TF $i$ over the same $L$ experiments. $S(i,j)$ is the control strength from TF $j$ to gene $i$. (B) The algebraic formulation of NCA.  $E$, $S$ and $A$ in (B) correspond to $E$, $S$ and $A$ in (A).}
\label{fig0}
\end{figure}

\subsection{Network Component Analysis (NCA) and Its Implementations }
The main principle of the NCA is to explain the expression of each gene as a linear combination of activities of  TFs regulating that gene weighted by their control strength that exerts on the gene. In case of NCA, the topology of the bipartite GRN is provided as a part of the input.  Formally, let $E \in \mathbb{R}^{N\times L}$ be the matrix of  expression data of $N$ genes in  $L$ experiments. NCA is a special case of a more general problem which is to express $E$ as
\begin{equation}\label{linear}
E=SA+\Gamma,
\end{equation}
where  $S \in \mathbb{R}^{N\times M}$ is a weighted adjacency matrix of the bipartite GRN 
$\mathcal{G}(\textsc{\textit{TF}}, \textsc{\textit{TG}}, \mathcal{S})$ such that  the  edges of  $\mathcal{G}$ in the edge set $\mathcal{S}$ connect transcription factors in the $M$ element set $\textsc{\textit{TF}}$ to  target genes from the $N$ element set $\textsc{\textit{TG}}$. Specifically, for target gene $i$ and transcription factor $j$,  weight $S(i,j)$ defines the control strength that transcription factor $j$ exerts on gene $i$.  The rows of  $A \in \mathbb{R}^{M\times L}$ represent the (hidden) TF activities of each TF over the set of experiments,  and $\Gamma \in \mathbb{R}^{N\times M}$ represents the noise (Fig.~\ref{fig0}).

Many mathematical techniques, such as principle component analysis (PCA), independent component analysis (ICA), non-negative matrix factorization (NMF)~\cite{Lee2001} and sparse coding (SP)~\cite{Mairal2009}, can be used to determine the decomposition of $E$ specified in~\eqref{linear} (for NMF, $E$ needs to be normalized to a non-negative matrix). However, PCA and NMF~\cite{Wang2013} are unable to find a decomposition of $E$ when $M>L$ (i.e. the number of TFs is larger than the number of experiments). Moreover, PCA and ICA hinge on the assumptions of orthogonality and independence between the signals, which may not hold for TF activities (rows of $A$). In addition, all of them can not utilize the prior knowledge of the GRN $\mathcal{G}$. In contrast, NCA~\cite{Liao2003,Tran2005,Galbraith2006,Chang2008,Jacklin2012,Noor2013} can deal with the situation when $M> L$, make no assumptions on TF activities and is able to take full advantage of the prior knowledge of the GRN $\mathcal{G}$. Specifically, NCA aims to uncover the matrix $A$ describing the hidden regulatory activities of TFs and matrix $S$ describing control strengths of each TF on target genes assuming that the structure $S_0$ (unweighted adjacency matrix) of the underlining GRN $\mathcal{G}_0=(\textsc{\textit{TF}}, \textsc{\textit{TG}}, \mathcal{S}_0)$ is known. That is only the entries of $S$ that correspond to edges in $\mathcal{S}_0$ can be non-zero (formally $ \textup{Support}(S) = \textup{Support}(S_0)$, where $\textup{Support}(S)$ denotes the support of $S$, i.e. the positions of its non-zero entries.). Thus NCA recovers the TF activities $A$ and their control strength $S$ by solving the following optimization problem with only the expression data $E$ and the structure $S_0$ of $\mathcal{G}_0$ as inputs.
\begin{equation}\label{nca}
\begin{aligned}
\min_{S,A}\ &\dfrac{1}{2}\left \| E - SA \right \|_F^2\\
       s.t.\ & \textup{Support}(S) = \textup{Support}(S_0),\\
           \ & \left \| S\right \|_{\infty} \leq a, \ \left \| A\right \|_{\infty} \leq b,
\end{aligned}
\end{equation}
where $\left \| S\right \|_{\infty} =\max_{i,j} |S(i,j)|$. The first constraint in the above formulation restricts the structure of the regulatory network $\mathcal{G}$ represented by matrix $S$ to be exactly the same to the input regulatory network $\mathcal{G}_0$. And the rest of the constraints aim to ensure the elements in $A$ and $S$ remain within the domain of biologically sensible values.  

The first method~\cite{Liao2003} to solve problem~\eqref{nca} can provide a unique solution only if the following conditions are met: (i) the matrix $S$ should have full-column rank; (ii) each column of $S$ should have at least $M-1$ zeros; (iii) the matrix $A$ should have full row-rank. Under these conditions, $S$ and $A$ are estimated using an iterative two-step least-squares algorithm~\cite{Liao2003}. Tran et al.~\cite{Tran2005} expanded NCA by allowing the specification of the zero pattern of $A$ as well as $S$. Galbraith et al.~\cite{Galbraith2006} modified the NCA method by revising the third criterion for NCA which cannot be tested before solving the problem. Chang et al.~\cite{Chang2008} treated NCA as an unconstrained optimization problem and employed singular value decomposition (SVD) to find a closed form solution for $S$ without time-consuming iterations. Jacklin et al.~\cite{Jacklin2012} also proposed a non-iterative algorithm for NCA resorting to convex optimization methods. All these methods are vulnerable to the presence of small number of outliners in expression data. To deal with these outliers, Noor et al.~\cite{Noor2013} proposed ROBust Network Component Analysis (ROBNCA) where an additional sparse matrix was used for explicitly modeling the outliers.

\subsection{The Formulation of the Optimization Problem Behind  \reNCA\ }

Independently of numerous variants of the NCA, the assumption that the GRN must be known in advance is a significant drawback of the NCA method. \reNCA relaxes this restriction  under the assumption that  a prior regulatory network that is not too far from the underlining true regulatory network is given. Therefore, it is possible to recover the underlining regulatory network by limited changes to the  prior network. Note that this is a very reasonable assumption in many practical applications where the prior network could come from a related organism or a related tissue or even from the same organism but without sufficient data.  Additionally, to guide the network reconstruction, we assume that genes with highly correlated expression are likely to be regulated by the same TFs. The correlations between genes can be encoded in the gene correlation network $\mathcal{G}^E$ constructed based on gene expression data $E$. Thus in the new optimization problem~\eqref{eql2} we remove the constraint that the structure of the network is fixed ($\textup{Support}(S) = \textup{Support}(S_0))$ but introduce a penalty term that limits the number of added and removed edges with respect to the prior network, along with the terms encouraging consistent treatment of co-expressed genes, and network sparsity. We devote the rest of this subsection to explaining the roles of the added terms.  
\begin{equation}\label{eql2}
\begin{aligned}
\min_{S,A}\ &\dfrac{1}{2}\left \| E - SA \right \|_F^2+\lambda \left (\left \| S_0 \right \|_0-\left \| S\odot S_0 \right \|_0+\left \| S \odot \bar{S}_0 \right \|_0 \right ) + \kappa \textup{tr}(S^TLS) + \eta \left \|S\right \|_0 + \xi \left \|S\right \|_F^2 + \mu \left \| A\right \|_F^2\\
       s.t.\ & \left \| S\right \|_{\infty} \leq a, \ \left \| A\right \|_{\infty} \leq b.
\end{aligned}
\end{equation}
where $\lambda, \kappa, \eta, \xi, \mu $ are the parameters controlling the strength of the corresponding terms. 

The term controlled by $\lambda$ restricts the number of edge changes. Here $\bar{S}_0$ is the adjacency matrix of the complement graph of $\mathcal{G}_0$ and therefore $\bar{S}_0+S_0=\textbf{1}_{N\times M}$ and $\left \|X\right\|_0$ is the $\ell_0$ norm that computes the number of non-zero entries in $X$. $\odot$ is the Hadamard product. We note that $\left \| S_0 \right \|_0-\left \| S\odot S_0 \right \|_0$ denotes the exact number of regulations removed from $\mathcal{G}_0$ and $\left \| S \odot \bar{S}_0\right \|_0$ is the number of regulations added to the prior network $\mathcal{G}_0$. $\lambda$ controls the change in topology of the regulatory network. Larger $\lambda$  indicates that only small number of edges can be added and removed controlling how far our predicted network $\mathcal{G}$ is from the prior network $\mathcal{G}_0$.

The term controlled by $\kappa$ (the graph embedding term~\cite{Belkin2003})  encourages $S(i,k)$ and $S(j,k)$ to have similar control strength if genes $i$ and $j$ are correlated. In Supplementary Materials~\ref{sm1} we provide derivations demonstrating that: 
\begin{equation}\label{graphembed}
\begin{aligned}
&\dfrac{1}{2}\sum_{i,j} \sum_k W(i,j) \left (S(i,k) - S(j,k) \right )^2=\textup{tr}(S^TLS),
\end{aligned}
\end{equation}
where $\textup{tr}()$ is the trace of a matrix and $W$ and $L$ are the adjacency matrix and  the Laplacian matrix of the correlation network $\mathcal{G}^E$, respectively.

The term of equation~\eqref{eql2} that is controlled by  parameter $\eta$ encourages sparsity of the final network (note that $\ell_0$ norm computes the number of non-zero elements). However we note that there might exist correlations between TF activities (rows of $A$), which would imply relations between TFs and enforcing the  sparsity might weaken them. This means that, for a gene, only one TF can be selected from a group of TFs whose activities are highly correlated even though all TFs in the group regulate the gene. Therefore, we have an additional term (controlled by parameter $\xi$) using Frobenius norm to encourage that all regulating TFs have non-zero values in $S$.  For the reader familiar with the elastic net model, we point that $\eta \left \|S\right \|_0 + \xi \left \|S\right \|_F^2$  is analogous  to $\ell_1$ elastic net~\cite{Zou2005}, and we can refer to it as $\ell_0$ elastic net.

Finally, the last term controlled by the variable $\mu$ enforces smoothness of  activities in $A$ by avoiding many elements in $A$ reach to the limit $\{-b, b\}$.

After some linear algebra (Supplementary Materials~\ref{sma2}), we obtain our final formulation as follow. We require $\eta-\lambda\geq 0$, otherwise the above formulation would preserve all regulations in $\mathcal{G}_0$. 
\begin{equation}\label{eqf}
\begin{aligned}
\min_{S,A}\ &\dfrac{1}{2}\left \| E - SA \right \|_F^2+ \left (\eta-\lambda\right )\left \| S\odot S_0 \right \|_0+\left ( \eta+\lambda \right )\left \| S \odot \bar{S}_0 \right \|_0  + \kappa \textup{tr}(S^TLS) + \xi \left \|S\right \|_F^2 + \mu \left \| A\right \|_F^2\\
       s.t.\ & \left \| S\right \|_{\infty} \leq a, \ \left \| A\right \|_{\infty} \leq b.
\end{aligned}
\end{equation}

\subsection{Solving the Optimization Problem Underlying the \reNCA Algorithm }

Our algorithm to solve  optimization problem~\eqref{eqf} relies on the recently proposed proximal alternative linearized maximization (PALM)~\cite{Bolte2014} algorithm. 
The PALM method can solve a general optimization problem formulated as

\begin{equation}\label{obj6}
\begin{aligned}
\min_{}: & \ H(S,A)= F(S,A) + \Phi(S) + \Psi(A) \ \textup{over} \ S\in \Upsilon, A \in \Omega,
\end{aligned}
\end{equation} 
where $F(S,A)$ has to be smooth but $\Phi(S)$ and $\Psi(A)$ do not need to have the convexity and smoothness properties. $\Upsilon$ and $\Omega$ are constraint sets for $S$ and $A$, respectively. The PALM algorithm alternatively applies technique known as proximal forward-backward scheme to both $S$ and $A$. Specifically, at iteration $k$, the proximal forward-backward mappings of $\Phi(S)$ and $\Psi(A)$ on $S\in \Upsilon$ and $A \in \Omega$ for given $S^k$ and $A^k$ are the solutions for the following sub-problems, respectively. 
\begin{subequations}
\begin{align}
&S^{k+1} \in \arg \min_{S\in \Upsilon} \left \{ \big \langle S -S^k, \ \nabla_S F(S^k,A^k) \big \rangle  + \frac{c^k}{2}\left \| S- S^k \right \|_F^2 + \Phi\left ( S \right ) \right \};\label{subeq1111}\\
&A^{k+1} \in \arg \min_{A \in \Omega} \left \{ \big \langle A-A^k, \ \nabla_A F(S^{k+1},A^k) \big \rangle + \frac{d^k}{2}\left \| A- A^k \right \|_F^2 + \Psi\left ( A \right ) \right \} \label{subeq2111},
\end{align}
\end{subequations}
where $\big \langle X, Y \big \rangle = \textup{tr}(X^TY)$, $c^k$ and $d^k$ are positive real numbers and $\nabla_S F(S^k,A^k)$ is the derivative of $F(S,A^k)$ with respect to $S$ at point $S^k$ for fixed $A^k$ and $\nabla_A F(S^{k+1},A^k)$ is the derivative of $F(S^{k+1},A)$ with respect to $A$ at point $A^k$ for fixed $S^{k+1}$. It has been proven that the sequence $\left \{ (S^k, A^k)\right \}_{k\in\mathbb{N}}$ generated by PALM converges to a critical point when it is bounded~\cite{Bolte2014}. 

Casting our optimization problem~\eqref{eqf} into the PALM algorithm framework~\eqref{obj6} introduced above, we have $F(S,A) := \dfrac{1}{2}\left \| E - SA\right \|_F^2+\kappa \textup{tr}(S^TLS)$,  $\Psi(A) :=  \mu \left \| A\right \|_F^2$ and $\Phi(S) := \left (\eta+\lambda\right )\left \| \bar{\mathbb{S}}_0  \odot S\right \|_0 +  \left (\eta - \lambda\right ) \left \| \mathbb{S}_0 \odot  S\right \|_0 + \xi \left \| S\right \|_F^2$. The constraint sets  $\Upsilon$ and $\Omega$ are respectively $\Upsilon = \left \{ S \ | \left \| S\right \|_{\infty}\leq a \right \}$ and $\Omega = \left \{ A \ | \left \| A\right \|_{\infty}\leq b \right \}$. We note that $F(S,A)$, $\Psi(A)$ and $\Phi(S)$ satisfy the requirements of the PALM algorithm. Namely, $F(S,A)$ is smooth, $\Psi(A)$ is convex and smooth but, as allowed in the PALM approach,  $\Phi(S)$ is  non-convex and non-smooth.
Hence, we can apply the PALM algorithm to our problem as long as we can efficiently solve the proximal forward-backward mappings for our specific $\Phi(S)$ and $\Psi(A)$. Proving that we can actually do it is mathematically  the most challenging component  of the development of the method. Due to technicality of the derivations we leave most of them to the supplement and in what follows we only point to the most critical components of the argument.  

It is easy to confirm that the \reNCA problem~\eqref{eqf} can be solved by alternatively applying the following proximal forward-backward mappings~\eqref{subeq1} and~\eqref{subeq2}, which are derived from~\eqref{subeq1111} and~\eqref{subeq2111} by casting our specific $F(S,A)$, $\Phi(S)$, $\Psi(A)$, $\Upsilon$ and $\Omega$ and some linear algebra (derivations can be found in Supplementary Materials~\ref{smpxd}).:
\begin{subequations}
\begin{align}
&S^{k+1} \in \arg \min_{\left \| S\right \|_{\infty}\leq a} \left \{ \Phi\left ( S \right ) + \frac{c^k}{2}\left \| S- U^k \right \|_F^2 \right \};\label{subeq1}\\
&A^{k+1} \in \arg \min_{\left \| A\right \|_{\infty}\leq b} \left \{ \Psi\left ( A \right ) + \frac{d^k}{2}\left \| A- V^k \right \|_F^2 \right \} \label{subeq2},
\end{align}
\end{subequations}
where 
\begin{equation}\label{dev}
U^k = S^k - \frac{1}{c^k}\nabla_S F(S^k,A^k)  \ \textup{and}\ V^k = A^k - \frac{1}{d^k} \nabla_A F(S^{k+1},A^k).
\end{equation}
The derivatives $\nabla_S F(S^k,A^k)$ and $\nabla_A F(S^{k+1},A^k)$ can be computed by
\begin{equation}\label{dev1}
\nabla_S F(S^k,A^k) = (S^kA^k-E)(A^k)^T+2\kappa LS^k \ \textup{and}\  \nabla_A F(S^{k+1},A^k) = (S^{k+1})^T(S^{k+1}A^k-E),
\end{equation}
which are Lipschitz continuous with $L(A^k)=\left \|A^k(A^k)^T\right \|_F+2\kappa \left \|L\right \|_F$ and $L(S^{k+1})=\left \|(S^{k+1})^TS^{k+1} \right \|_F$ as Lipschitz constants, respectively. As suggested by~\cite{Bolte2014}, we set $c^k= \max\left \{v, \ L(A^k)\right \}, \ v>0$ and $d^k= \left \{v, \ L(S^{k+1})\right \}, \ v>0$ to make sure the formulas in~\eqref{dev} are well defined. 

The closed form solution of the proximal forward-backward mapping~\eqref{subeq1} can be obtained based on Proposition~\ref{lemma2}, the Proximal Mapping of $\ell_0$ Elastic Net Under $\| \|_{\infty}$ Constraint Proposition, and its corollary (Corollary~\ref{pro01}). The proposition and the corollary and their proofs can be found in the Supplementary Materials~\ref{sm1b1} and~\ref{sm2}. We emphasize that Proposition~\ref{lemma2} provides the closed form solution for the proximal mapping of $\ell_0$ elastic net under $\| \|_{\infty}$ constraint and thus it has broader applications to diverse feature selection approaches~\cite{Aben2016,Das2015}.

With the help of Proposition~\ref{lemma2} and Corollary~\ref{pro01}, \eqref{subeq1} can be efficiently computed by~\eqref{px1}. And~\eqref{subeq2} can be computed by~\eqref{px2}.     
\begin{subequations}
\begin{align}
&S^{k+1} \in \textup{prox}_{\left\|\cdot\right\|_{\infty}\leq a}\left (U^k,\frac{2(\eta+\lambda)}{c^k},  \frac{2(\eta-\lambda)}{c^k}, \frac{2\xi}{c^k}\right );\label{px1}\\
&A^{k+1} = \mathbb{P}_{\left \| \cdot \right \|_{\infty}\leq b}\left ( \dfrac{1}{1+\frac{2\mu}{d^k}}V^k \right ),\label{px2}
\end{align}
\end{subequations}
where the definitions of $\textup{prox}_{\left\|\cdot\right\|_{\infty}\leq a}(\cdot)$ and $\mathbb{P}_{\left \| \cdot \right \|_{\infty}\leq b}(\cdot)$ can be found in Corollary~\ref{pro01} and Proposition~\ref{lemma2}, respectively. The derivations of~\eqref{px1} and~\eqref{px2} can be found in the Supplementary Materials~\ref{sm3-1}.

We now have all the ingredients for our \reNCA algorithm. Hence, we describe the \reNCA algorithm in Algorithm~\ref{alg1} in Supplementary Materials~\ref{sm2c}.  We note that the constraints for both $S$ and $A$ ($\left \| S\right \|_{\infty}\leq a$ and $\left \| A\right \|_{\infty}\leq b$) make sure that the sequence $\left \{ (S^k, A^k)\right \}_{k\in\mathbb{N}}$ is bounded.  Thus we state that the sequence produced by the \reNCA algorithm converges to a critical point of the optimization problem~\eqref{eqf}, which is described in Proposition~\ref{converg} in Supplementary Materials~\ref{sm3}.

\begin{figure}[htpb]
\includegraphics[width=16.5cm,height=5.8cm]{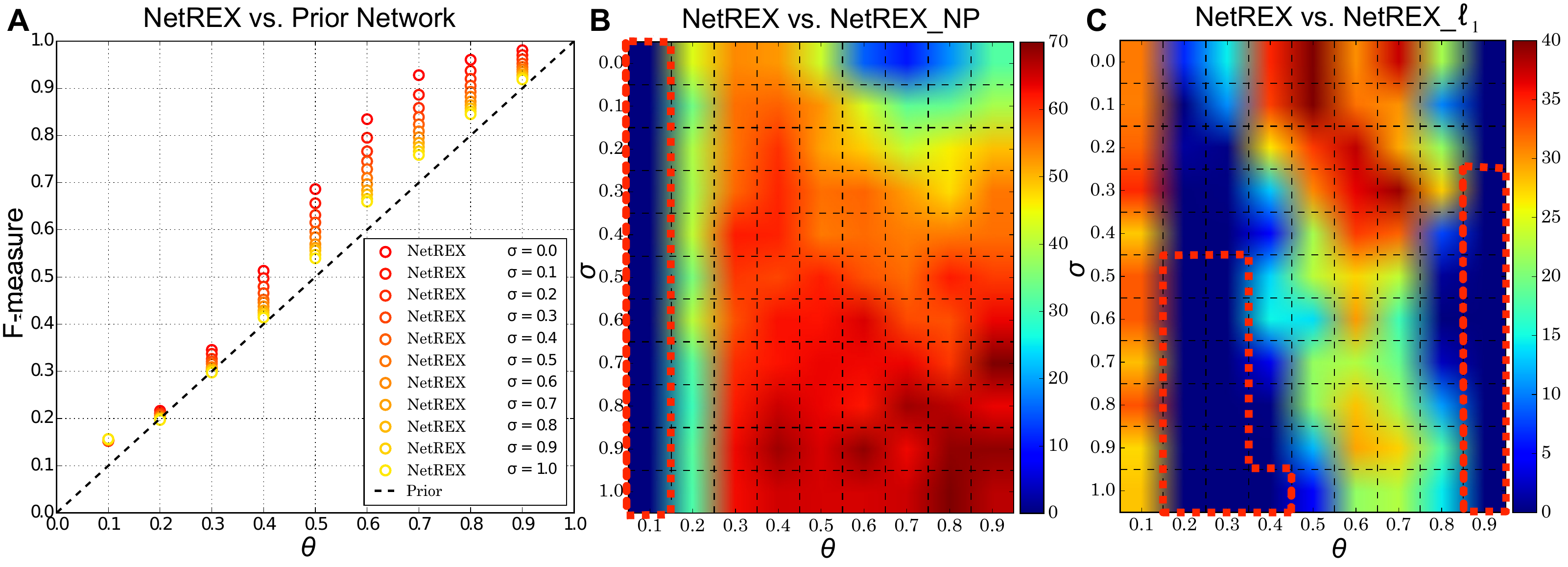}
\caption{\scriptsize (A) Comparison between F-measures of the networks predicted by \reNCA and   prior networks. The x-coordinate denotes percentages of true edges in  prior networks and the black dashed line denotes the F-measures of the prior networks. The circles are the average F-measures of the networks predicted by \reNCA under different $\sigma$ and $\theta$ over 50 random inputs. (B) Comparison between F-measures of the networks predicted by \reNCA and \reNCAs\_NP. The color in each dashed block indicates the $-\log{\textup{p-value}}$ for corresponding $(\sigma, \theta)$, where the p-values are obtained from one-sided paired t-test between F-measures of the compared algorithms. The warmer the color is, the significantly larger F-measures of the networks predicted by \reNCA are than \reNCAs\_NP. The red dashed line circles the $(\sigma, \theta)$ pairs where \reNCAs\_NP achieves larger F-measure at significant level 0.01. (C) Comparison between F-measures of the networks predicted by \reNCA and \reNCAs\_$\ell_1$. The color coding is the same as in the panel (B).} 
\label{fig1}
\end{figure}

\section{Validation and Experimental Results}

\subsection{Results on Simulated Data}
To validate our approach, we applied \reNCA to the simulated data generated based on linear model~\eqref{linear}. We first randomly generated the ground truth adjacency matrix $S$ of the regulatory network $\mathcal{G}(\textsc{\textit{TF}}, \textsc{\textit{TG}}, \mathcal{S})$ and TF activities $A$. Then, the simulated expression data was generated as following 
\begin{equation}\label{smd}
E(i,j) = \sum_p S(i,p)A(p,j) + \Gamma(i,j),
\end{equation}
where $\sum_p S(i,p)A(p,j)$ is the noiseless data arising from known $A$ and $S$ matrices and the noise $\Gamma(i,j)\sim\textbf{\textup{N}}(0, \sigma^2)$ obeys a normal distribution with 0 mean and $\sigma^2$ variance. We assigned the prior network $\mathcal{G}_0$ the same number of edges as the ground truth network $\mathcal{G}$ has but only $\theta$ percentage of edges in $\mathcal{G}_0$ are true edges. We can tune the difficulty of the network rewiring task by using different $\sigma$ and $\theta$.

We compared \reNCA with its two natural variants  on the simulation data. The first variant is \reNCAs\_NP (\reNCA with No edge Perturbation term) that has the same formulation as \reNCA but with $\lambda=0$. The difference between \reNCA and \reNCAs\_NP is that \reNCA penalizes the number of edges added and removed from the prior network but \reNCA\_NP does not. Here we want to mention that \reNCA\_NP and sparse coding have similar formulations (Supplementary Materials~\ref{smc1}). The other related algorithm in our comparison is  \reNCAs\_$\ell_1$, which estimates the $\ell_0$ norm in \reNCA using $\ell_1$ norm. We note, that substituting $\ell_0$ norm by $\ell_1$ norm makes the sub-problems convex and thus easier to solve.  The implementation of these two algorithms are introduced in Supplementary Materials~\ref{smc1}.

We evaluated the performance of the compared algorithms in terms of F-measure (defined in Supplementary Materials~\ref{smb1}). To avoid the effect of parameter selection, for each algorithm, under certain noise level $(\sigma, \theta)$, we first found its optimal parameters in terms of F-measure on one simulated data set through grid search. Then we ran the algorithm on another 50 randomly generated simulated data sets under the same $(\sigma, \theta)$ using its optimal parameters. We can further test whether one method is statistically better than another method under a specific noise level by computing the p-value from one-side paired t-test between their 50 paired F-measures.

The comparisons between \reNCA and others are shown in Fig.~\ref{fig1}. Fig.~\ref{fig1}A shows the comparison between networks predicted by \reNCA and the prior networks, in which we found a tendency that when the expression data is less noisy ($\sigma$ is small) and the prior network is closer to the ground truth ($\theta$ is large), the network predicted by \reNCA achieves higher F-measures than the prior networks. Additionally, we note that \reNCA exhibits, by a larger margin, higher F-measure comparing to the prior networks after $\theta \geq 0.3$. However, for $\theta < 0.3$ the networks predicted by \reNCA only marginally better than the prior network, which also implies that if we use random networks that do not have much overlap with the ground truth as the prior networks, we can not obtain promising results. The comparison between \reNCA and \reNCAs\_NP is displayed in Fig.~\ref{fig1}B. We note that \reNCA significantly outperforms \reNCAs\_NP after $\theta>0.1$. In Fig.~\ref{fig1}C, we observe that \reNCAs\_$\ell_1$ performs better in certain cases where the noise in the expression data is large ($\sigma$ is large) because $\ell_1$ norm is robust to noise. However, for most of the noise levels, \reNCA achieved significantly higher F-measures comparing to \reNCAs\_$\ell_1$ demonstrating that  $\ell_0$ norm is superior to $\ell_1$ norm on selecting sparse contributing components.

\vspace{-0.1cm}
\subsection{Results on Real Experimental Data from DrosDel Study}\vspace{-0.1cm}
Next we applied \reNCA to gene expression data in the adult female flies from 99 hemizygotic lines (deletion/+) of the \textit{Drosophila} deletion collection (DrosDel) project covering 68\% of chromosome 2L. Specifically, in each of the DrosDel lines, a different chromosomal fragment has been deleted leaving the organism with only one copy of genes for the deleted region~\cite{Ryder2007}. We used the network constructed in~\cite{Marbach2012} as the prior network, which is constructed through integrating diverse functional genomics data sets (such as transcription factor (TF) binding, evolutionary conserved sequence motifs and etc.) in a supervised learning framework. The data used in~\cite{Marbach2012} typically comes from cell lines and expression profiles of developmental stages. \reNCA predicted regulatory networks for the adult female flies. And we verified these networks using GO functional annotations~\cite{Blake2015}, physical protein-protein interactions (PPIs) and \textit{Drosophila} Doublesex transcription factor (DSX) target genes~\cite{Clough2014}.

We compared our predicted networks with the prior networks and the TF-Gene correlation networks that were built based on the Pearson coefficient between expression of TFs and genes using expression measurements in DrosDel data~\cite{Lee2016}.  \tcr{Furthermore, we compared with GENIE3~\cite{Huynh-Thu2010}, the best performer in the DREAM4 \textit{In Silico Multifactorial challenge}~\cite{Marbach2012a}, which predicts GRNs using only expression data. We also tried to compare with NMF based algorithms, which can not be applied because the dimensional requirement of NMF~\cite{Wang2013} is not satisfied (i.e. the number of TFs in the target GRN is larger than the number of hemizygotic lines in DrosDel data).}  To demonstrate performance of \reNCA under different parameters, and to choose the parameters in a manner that does not depend on tested data (GO annotations and PPIs), we developed a simple heuristic that ranks the models performance using a quality score based on the objective function~\eqref{eqf} (Supplementary Materials~\ref{smc3}). We then used top twenty models with respect to this ranking.  \tcr{In addition, since performance of the correlation based algorithm and GENIE3 might depend on different cut-offs, in the evaluation we showed the performance of the TF-Gene correlation networks and networks predicted by GENIE3 with different cutoffs (Supplementary Materials~\ref{smc3}).} Finally, we note that, unlike other networks in this comparison, the co-expression network is not a regulatory network, however, we embedded the information about this network in our objective function, thus we  need to show that the good performance of our method is not merely reflection of embedding this information in the objective function. All details of parameter setting used in the comparison are elaborated in the Supplementary Materials~\ref{smc3}.     

\vspace{-0.3cm}
\begin{figure}[htpb]
\centering 
\includegraphics[width=16.5cm]{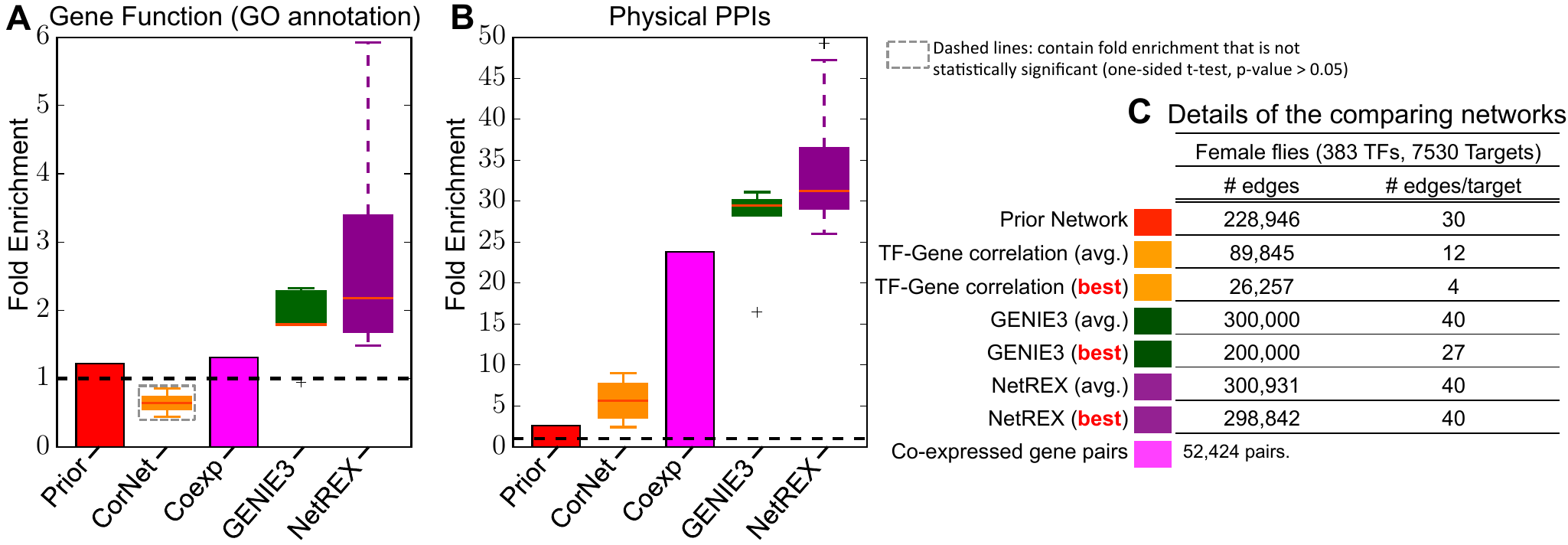}
\caption{\scriptsize Functional enrichment and the details of the compared networks. \tcr{For the correlation based algorithm, GENIE3 and \reNCAs, we showed the performance of these three algorithms under different parameters (details in Supplementary Materials~\ref{smc3}).}  (A) is the fold enrichment of GO annotations between co-regulated genes. (B) is the fold enrichment of physical interactions between co-regulated genes. (C) is the details of all compared networks. The best networks of each method are obtained considering the fold enrichment of both GO annotation and PPIs (Supplementary Materials~\ref{smc4}). The networks predicted by \reNCA show stronger enrichment for both gene functions and physical PPIs. }
\label{flydata}
\end{figure}
\vspace{-0.2cm}

\subsubsection{Functional Enrichment of the Predicted Regulatory Networks}
We assessed the biological relevance of the predicted regulatory networks through checking whether genes co-regulated by similar TFs exhibit similar functional properties. We used the measures (a brief review is in Supplementary Materials~\ref{smb2} and~\ref{smb3}) proposed in~\cite{Marbach2012} to evaluate the enrichment of co-regulated genes in terms of GO functional annotations~\cite{Blake2015} and experimentally derived physical Protein Protein interactions (PPIs) extracted from DroID database~\cite{Yu2008}, respectively. \tcr{In addition to the prior networks, the TF-Gene correlation networks and the networks predicted by GENIE3, we also computed the same measures for the co-expression network  inputted as the graph embedding term in \reNCAs.}

First, we examined whether co-regulated genes have similar GO annotations. The comparison results are illustrated in Fig.~\ref{flydata}A.  
\reNCA clearly outperforms all other approaches demonstrating  benefit of using both the prior network and the condition specific gene expression data.
       
Next, we evaluated whether physically interacting genes are tend to be co-regulated in the respective networks. Fig.~\ref{flydata}B shows the PPIs enrichment comparison. Using this enrichment as a measure of network quality as it has been proposed in~\cite{Marbach2012} we observed that \reNCA also outperforms all other methods. 

\begin{figure}[htpb]
\centering 
\includegraphics[width=16.5cm]{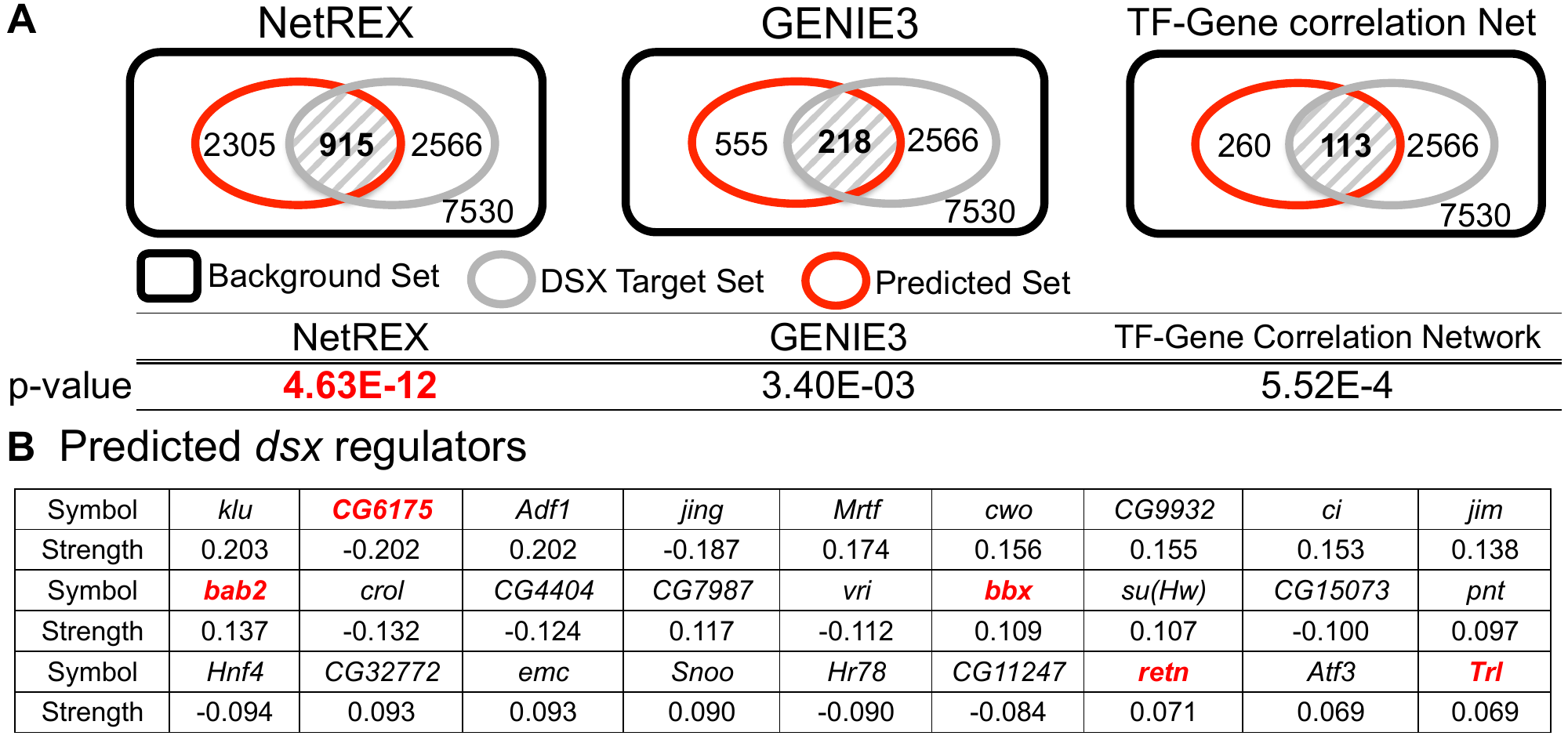}
\caption{\scriptsize (A) Comparison of agreement with DSX Targets using Venn diagrams with p-values.  \tcr{P-values are computed by the hyper geometric test. \reNCA achieved the lowest p-value. In the prior network, there are 2 target genes of DSX but none of them overlapped with the target set reported in~\cite{Clough2014}.}  (B) \textit{dsx} regulators predicted by \reNCA for female flies. We ranked the regulators based on the absolute values of their control strength. Genes in red have evidence in literature. }
\label{agreement}
\end{figure}

\subsubsection{Agreement with DSX Targets}
\tcr{To further validate the predictions obtained from different methods}, we concentrated on target genes of \textit{Drosophila} Doublesex transcription factor (DSX) which involves in the sex determination system as different isoforms in flies~\cite{Baker1989}. Recently, Clough et al.~\cite{Clough2014} reported a rich set of DSX targets based on a series of genome-wide experiments and analysis. Thus, we checked how well the predicted DSX targets of \reNCA are in agreement with genes identified in~\cite{Clough2014}. For each comparing method, we selected the network giving the best fold enrichment in both GO annotation and physical PPIs (Supplementary Materials~\ref{smc4}).  As shown in Fig.~\ref{agreement}A, \reNCA outperforms other methods and there is statistically significant overlap between targets predicted by \reNCA and targets inferred from experimental data (p-value are computed by the hyper geometric test).

\subsubsection{Prediction and Validation of DSX Regulators} 
We then used \reNCA to predict the regulators of \textit{doublesex} gene (\textit{dsx}).  Fig.~\ref{agreement}B illustrates the results of our prediction for \textit{dsx} in female flies.  Our predictions include multiple probable transcription factors that for example, \textit{Trl} (\textit{Trithorax-like}) gene encodes \textit{Drosophila} GAGA factor that has reported roles in sex chromosome dosage compensation~\cite{Greenberg2004}.   In addition, the \textit{retn} locus is required to repress male courtship behaviors in females~\cite{Ditch2005} while \textit{bab1} and \textit{bab2} genes have overlapping functions controlling sex-specific morphology~\cite{Couderc2002}.  We also observed many other predicted genes in the list have sex-specific expression, as like \textit{bbx} (\textit{boby sox}) and \textit{CG6175}, which demonstrate testis-specifically repressed, but ovary-expressed patterns~\cite{Robinson2013}.

\vspace{-0.1cm}
\section{Conclusion}\vspace{-0.1cm}
Regulatory networks embed key information needed for modeling and interpreting  experimental data.  Currently, regulatory networks are constructed by combining information from various tissues, stages, and conditions~\cite{Marbach2012,Turkarslan2015,Hecker2015} to obtain a context-agnostic network. However the importance of constructing tissue / stage specific  networks is now being increasingly recognized~\cite{Park2015,Greene2015}.  And constructing such network from scratch for every relevant tissue and/or condition is not realistic. In addition, data obtained from different tissues and conditions might, provide additional information that context specific analysis might not be sufficiently empowered to detect. For these reasons it is fundamental to be able to utilize regulatory networks constructed in context independent way as a starting point for context specific network construction. The \reNCA method proposed in this paper fulfills this critical need. Importantly, our construction is mathematically rigorous.   We proved convergence of the method and  validate its performance on both simulation data and real world data. The experimental results demonstrate that \reNCA is capable to recover the biological meaningful condition specific TF-gene regulatory networks.  

\section{Acknowledgments}
This work was supported by the Intramural Research Program of the National Institutes of Health, National Library of Medicine (YW, TMP, DYC) and  National Institute of Diabetes and Digestive and Kidney Diseases National (HK, BO). The authors thanks Roded Sharan, and the members of Przytycka and Oliver groups for helpful discussions.  


\bibliographystyle{unsrt}
\bibliography{library}

\newpage
\begin{center}
\huge{\textbf{Supplementary Materials I:}}\\
\LARGE{\textbf{Mathematical Derivations and Proofs}}
\end{center}

\renewcommand\thesection{\Alph{section}}
\setcounter{section}{0}


\section{Mathematical Derivations}
\subsection{The Graph Embedding}\label{sm1}
The derivation of the graph embedding term~\eqref{graphembed} is shown below.
\begin{equation}\label{eql}
\begin{aligned}
&\dfrac{1}{2}\sum_{i,j} \sum_k W(i,j) \left (S(i,k) - S(j,k) \right )^2\\
&=\dfrac{1}{2}\sum_{i,j} W(i,j)\left \| S(i,:)- S(j,:) \right \|_2^2\\
&=\dfrac{1}{2} \left ( \sum_{i,j} W(i,j)S(i,:)S(i,:)^T +  \sum_{i,j} W(i,j)S(j,:)S(j,:)^T - 2 \sum_{i,j} W(i,j)S(i,:)S(j,:)^T\right  )\\
&=\dfrac{1}{2} \left ( \sum_{i} D(i,i)S(i,:)S(i,:)^T+  \sum_{j} D(j,j)S(j,:)S(j,:)^T - 2 \sum_{i,j} W(i,j)S(i,:)S(j,:)^T\right  )\\
&=\dfrac{1}{2} \left ( 2\textup{tr}(S^TDS) - 2\textup{tr}(S^TWS)\right  )\\
&=\textup{tr}(S^TLS),
\end{aligned}
\end{equation}
where $\textup{tr}()$ is the trace of a matrix. $W$ is the adjacency matrix for the gene co-expression network $\mathcal{G}^E$ and $D$ is a diagonal matrix with the degree of every gene on its diagonal. $L$ is the Laplacian matrix of $\mathcal{G}^E$. 

\subsection{Derivation from \eqref{eql2} to \eqref{eqf}}\label{sma2}
The key of the derivation is the following transformation.
\begin{equation}\label{der1}
\begin{aligned}
&\lambda \left (\left \| S_0 \right \|_0-\left \|  S\odot S_0 \right \|_0+\left \| S \odot \bar{S}_0 \right \|_0 \right ) + \eta \left \|S\right \|_0 \\ 
=&\lambda \left (\left \| S_0 \right \|_0-\left \|  S\odot S_0 \right \|_0+\left \| S \odot \bar{S}_0 \right \|_0 \right ) + \eta \left \|S\odot \textbf{1} \right \|_0 \\
=&\lambda \left (\left \| S_0 \right \|_0-\left \|  S\odot S_0 \right \|_0+\left \| S \odot \bar{S}_0 \right \|_0 \right ) + \eta \left \|S\odot (S_0+\bar{S}_0) \right \|_0 \\
=&\lambda \left (\left \| S_0 \right \|_0-\left \|  S\odot S_0 \right \|_0+\left \| S \odot \bar{S}_0 \right \|_0 \right ) + \eta \left \|S\odot S_0+ S\odot \bar{S}_0 \right \|_0 \\
=&\lambda \left (\left \| S_0 \right \|_0-\left \|  S\odot S_0 \right \|_0+\left \| S \odot \bar{S}_0 \right \|_0 \right ) + \eta \left ( \left \|S\odot S_0\right \|_0 + \left \|S\odot \bar{S}_0 \right \|_0 \right ) \\
=&\lambda \left \| S_0 \right \|_0 + \left (\eta-\lambda\right )\left \| S\odot S_0 \right \|_0+\left ( \eta+\lambda \right )\left \| S \odot \bar{S}_0 \right \|_0. 
\end{aligned}
\end{equation}
In the above derivation, we use the fact that $\textbf{1} = S_0+\bar{S}_0$ and the property that $\left \|S\odot S_0+ S\odot \bar{S}_0 \right \|_0 = \left \|S\odot S_0\right \|_0 + \left \|S\odot \bar{S}_0 \right \|_0$ because $S_0$ and $\bar{S}_0$ are complementary to each other. Replacing the corresponding term in~\eqref{eql2} with the last line of~\eqref{der1}, we have 
\begin{equation}
\begin{aligned}
\min_{S,A}\ &\dfrac{1}{2}\left \| E - SA \right \|_F^2+ \lambda \left \| S_0 \right \|_0  + \left (\eta-\lambda\right )\left \| S\odot S_0 \right \|_0+\left ( \eta+\lambda \right )\left \| S \odot \bar{S}_0 \right \|_0  + \kappa \textup{tr}(S^TLS) + \xi \left \|S\right \|_F^2 + \mu \left \| A\right \|_F^2\\
       s.t.\ & \left \| S\right \|_{\infty} \leq a, \ \left \| A\right \|_{\infty} \leq b.   
\end{aligned}
\end{equation}
Because $\left \| S_0 \right \|_0$ is a constant number, we disregard it and our formulation becomes to~\eqref{eqfb} that is exactly the same to~\eqref{eqf}.
\begin{equation}\label{eqfb}
\begin{aligned}
\min_{S,A}\ &\dfrac{1}{2}\left \| E - SA \right \|_F^2+  \left (\eta-\lambda\right )\left \| S\odot S_0 \right \|_0+\left ( \eta+\lambda \right )\left \| S \odot \bar{S}_0 \right \|_0  + \kappa \textup{tr}(S^TLS) + \xi \left \|S\right \|_F^2 + \mu \left \| A\right \|_F^2\\
       s.t.\ & \left \| S\right \|_{\infty} \leq a, \ \left \| A\right \|_{\infty} \leq b.   
\end{aligned}
\end{equation}

\subsection{Derivations from~\eqref{subeq1111} and~\eqref{subeq2111} to~\eqref{subeq1} and~\eqref{subeq2}}\label{smpxd}
The derivation is based on completing the square technique. We simply show the derivation from~\eqref{subeq1111} to~\eqref{subeq1} and the derivation from~\eqref{subeq2111} to~\eqref{subeq2} is the same.

\begin{equation}
\begin{aligned}
S^{k+1} &\in \arg \min_{S\in \Upsilon} \left \{ \big \langle S -S^k, \ \nabla_S F(S^k,A^k) \big \rangle  + \frac{c^k}{2}\left \| S- S^k \right \|_F^2 + \Phi\left ( S \right ) \right \}\\
& \in \arg \min_{S\in \Upsilon} \left \{ \big \langle S -S^k, \ \nabla_S F(S^k,A^k) \big \rangle  + \frac{c^k}{2}\big \langle S- S^k, S- S^k\big \rangle + \Phi\left ( S \right ) \right \}\\
& \in \arg \min_{S\in \Upsilon} \left \{ \textup{tr}(S^T\nabla_S F(S^k,A^k))  + \frac{c^k}{2}\textup{tr}(S^TS- 2S^TS^k + (S^k)^TS^k) + \Phi\left ( S \right ) \right \}\\
& \in \arg \min_{S\in \Upsilon} \left \{ \Phi\left ( S \right ) + \frac{c^k}{2}\left \| S- (S^k - \frac{1}{c^k}\nabla_S F(S^k,A^k) ) \right \|_F^2 \right \}\\
& \in \arg \min_{\left \| S\right \|_{\infty}\leq a} \left \{ \Phi\left ( S \right ) + \frac{c^k}{2}\left \| S- (S^k - \frac{1}{c^k}\nabla_S F(S^k,A^k) ) \right \|_F^2 \right \}.
\end{aligned}
\end{equation}
The second last line is obtained based on completing the square technique and disregarding a term that is consist of the constants $S^k$ and $\nabla_S F(S^k,A^k)$. The last line is obtained by putting in the specific constraint for $S$.

\subsection{Derivations of the Proximal Mappings used in the \reNCA Algorithm}\label{sm3-1}
The derivation for~\eqref{px1} using Corollary~\ref{pro01} is shown below. 
\begin{equation}
\begin{aligned}
S^{k+1} & \in \arg \min_{\left \| S\right \|_{\infty}\leq a} \left \{ \Phi\left ( S \right ) + \frac{c^k}{2}\left \| S- U^k \right \|_F^2 \right \}\\
&= \arg \min_{\left \| S\right \|_{\infty}\leq a} \left \{ \frac{2(\eta+\lambda)}{c^k}\left \| \bar{\mathbb{S}}_0  \odot S\right \|_0 +  \frac{2(\eta-\lambda)}{c^k} \left \| \mathbb{S}_0 \odot  S\right \|_0 + \frac{2\xi}{c^k} \left \| S\right \|_F^2 + \left \| S- U^k \right \|_F^2 \right \}\\
&= \textup{prox}_{\left\|\cdot\right\|_{\infty}\leq a}\left (U^k,\frac{2(\eta+\lambda)}{c^k},  \frac{2(\eta-\lambda)}{c^k}, \frac{2\xi}{c^k}\right ).
\end{aligned}
\end{equation}
The derivation for~\eqref{px2} is shown below. 
\begin{equation}\label{proxa0}
\begin{aligned}
A^{k+1} &= \arg \min_{\left \| A\right \|_{\infty}\leq b} \left \{ \Psi\left ( A \right ) + \frac{d^k}{2}\left \| A- V^k \right \|_F^2 \right \}\\
             &= \arg \min_{\left \| A\right \|_{\infty}\leq b} \left \{ \frac{2\mu}{d^k} \left \| A\right \|_F^2 + \left \| A- V^k \right \|_F^2 \right \}\\
             &=\mathbb{P}_{\left \| \cdot \right \|_{\infty}\leq b}\left ( \dfrac{1}{1+\frac{2\mu}{d^k}}V^k \right ).
\end{aligned}
\end{equation}

\section{Propositions, Corollary, and Their Proofs}

\subsection{The Proposition for the Proximal Mapping of $\ell_0$ Elastic Net and Its Proof}\label{sm1b1}

\begin{propname}[\textbf{Proximal Mapping of $\ell_0$ Elastic Net Under $\| \|_{\infty}$ Constrain Proposition}]\label{lemma2}
For given $Y\in \mathbb{R}^{m\times n}$, the proximal mapping of $\ell_0$ elastic net under $\| \|_{\infty}$ norm Constrain is
 \begin{equation}
 \begin{aligned}
  \arg \min_{\left \|X\right \|_{\infty}\leq C}\left \{ \left \| Y-X \right \|_F^2 + b \left \|X\right \|_F^2 + c^2\left \| X \right \|_0 \right \}=T_{\frac{c}{\sqrt{b+1}}}\left (\mathbb{P}_{\left \| \cdot \right \|_{\infty} \leq C}\left (\dfrac{Y}{b+1} \right )\right ),
 \end{aligned}
 \end{equation} 
where the projection operator $\mathbb{P}_{\left \| \cdot \right \|_{\infty}\leq C}(\cdot)$ is defined in~\eqref{proj0}
\begin{equation}\label{proj0}
\mathbb{P}_{\left \| \cdot \right \|_{\infty}\leq C}(Y) := \arg \min\left \{ \left \| Y-X \right \|_F^2:  \left \| X \right \|_{\infty}\leq C\right \}= \textup{sign}(Y)\odot \max\left \{|Y|,  C\right \},
\end{equation}
where the $||$, $\textup{sign}()$ and $\max\{\}$ operations are taken component-wise. And the hard-thresholding operator $T_c(\cdot )$ is defined as
 \begin{equation}\label{hts}
 \begin{aligned}
 T_c(Y) := \arg \min_X\left \{ \left \| Y-X \right \|_F^2 + c^2\left \| X \right \|_0 \right \},
 \end{aligned}
 \end{equation} 
where $Y\in \mathbb{R}^{m\times n}$ is any given matrix and $T_c: \mathbb{R}^{m\times n} \rightarrow \mathbb{R}^{m\times n}$ is a component-wise mapping that can be explicitly wrote out
\begin{equation}
\left (T_c(Y)\right )(i,j)=\left\{\begin{matrix}
Y(i,j), & \textup{if} \ \left | Y(i,j) \right |> c;\\ 
\left \{ 0, c \right \}, & \textup{if} \ \left | Y(i,j) \right |= c; \\ 
0, & o.w..
\end{matrix}\right.
\end{equation}
\end{propname}

\begin{proof}
 \begin{equation}\label{hts}
 \begin{aligned}
& \arg \min_{\left \|X\right \|_{\infty}\leq C}\left \{ \left \| Y-X \right \|_F^2 + b \left \|X\right \|_F^2 + c^2\left \| X \right \|_0 \right \}\\
                   =& \arg \min_{\left \|X\right \|_{\infty}\leq C}\left \{ (b+1)\left \| \dfrac{Y}{b+1}-X \right \|_F^2  + c^2\left \| X \right \|_0 \right \}\\
                   =& \arg \min_{\left \|X\right \|_{\infty}\leq C}\left \{ \left \| \dfrac{Y}{b+1}-X \right \|_F^2  + (\dfrac{c}{\sqrt{b+1}})^2\left \| X \right \|_0 \right \}\\
                   =&T_{\frac{c}{\sqrt{b+1}}}\left (\mathbb{P}_{\left \| \cdot \right \|_{\infty} \leq C}\left (\dfrac{Y}{b+1} \right )\right ).
 \end{aligned}
 \end{equation} 
 The last equality is derived based on Lemma~\ref{lemma1}.
 \end{proof}

\begin{lemma}\label{lemma1}
Let $U\in \mathbb{R}^{m\times n}$, then
\begin{equation}
\begin{aligned}
\arg \min\left \{\left \| U - X \right \|_F^2 + c^2 \left \| X\right \|_0 : \left \| X\right \|_{\infty} \leq C\right \}=T_c\left (\mathbb{P}_{\left \| \cdot \right \|_{\infty} \leq C}(U) \right ).
\end{aligned}
\end{equation}
\end{lemma}

\begin{proof}
For given $U\in \mathbb{R}^{m\times n}$, let us introduce the following notations
\begin{equation}
\left \| X\right \|^2_{+} = \sum_{(i,j)\in \mathcal{I}^+} X(i,j)^2 \ \textup{and}\ \left \| X\right \|^2_{-} = \sum_{(i,j)\in \mathcal{I}^-} X(i,j)^2,
\end{equation}
where
\begin{equation}
\mathcal{I}^+ = \left \{(i,j) \in \{1, ..., m\} \times \{1, ..., n\}: |U(i,j)| \leq C \right \}
\end{equation}
and
\begin{equation}
\mathcal{I}^- = \left \{(i,j) \in \{1, ..., m\} \times \{1, ..., n\}: |U(i,j)| > C \right \}
\end{equation}
The following observations hold 
\begin{equation}
\begin{aligned}
&(i) \left \| X\right \|_F^2 = \left \|X \right \|_+^2 + \left \|X \right \|_-^2 \\
&(ii) \left \| X-U\right \|_+^2 + \left \| X-C\right \|_-^2 = \left \| X - \mathbb{P}_{\left \| \cdot \right \|_{\infty} \leq C}(U) \right \|_F^2\\
&(iii) \left \| X-C \right \|_-^2 =0 \Leftrightarrow X(i,j)=C \  \forall (i,j)\in \mathcal{I}^-
\end{aligned}
\end{equation}
where the second observation follows from observation $(i)$ and the fact that $\left (\mathbb{P}_{\left \| \cdot \right \|_{\infty} \leq C}(U)\right )(i,j) = U(i,j)$ for $(i,j)\in \mathcal{I}^+$ and $\left (\mathbb{P}_{\left \| \cdot \right \|_{\infty} \leq C}(U)\right )(i,j) = C$ for $(i,j)\in \mathcal{I}^-$.

Based on the above facts, we have that $\bar{X} \in \textup{prox}_{\left\| \cdot \right \|_0}^{\left \| \cdot \right \|_{\infty} \leq C}(U, c)$ if and only if
\begin{subequations}
\begin{align}
\bar{X} &\in \arg \min\left \{\left \| U - X \right \|_F^2 + c^2 \left \| X\right \|_0 : \left \| X\right \|_{\infty} \leq C\right \} \label{c1sub1}\\
            &=\arg \min\left \{\left \| U - X \right \|_+^2 + \left \| U - X \right \|_-^2 + c^2 \left \| X\right \|_0 : \left \| X\right \|_{\infty} \leq C\right \} \label{c1sub2} \\
            &=\arg \min\left \{\left \| U - X \right \|_+^2  + c^2 \left \| X\right \|_0 : X(i,j)=C \  \forall (i,j)\in \mathcal{I}^-, \ \left \| X\right \|_{\infty} \leq C\right \} \label{c1sub3},
\end{align}
\end{subequations}
where the last equality follows that fact that the solution of~\eqref{c1sub3} is also the solutions of~\eqref{c1sub2}, while the converse follows by a simple contradiction argument. Furthermore, one can find that the constraint $ \left \| X\right \|_{\infty} \leq C$ can be removed without affecting the optimal solution of the problem. Therefore, applying observation $(ii)$ and $(iii)$, we obtain
\begin{equation}
\begin{aligned}
\bar{X} &\in \arg \min\left \{\left \| U - X \right \|_+^2  + c^2 \left \| X\right \|_0 : \left \| X-C \right \|_-^2 =0\right \}\\
            &= \arg \min\left \{\left \| U - X \right \|_+^2  +  \left \| X-C \right \|_-^2 + c^2 \left \| X\right \|_0 \right \}\\
            &= \arg \min\left \{\left \| X - \mathbb{P}_{\left \| \cdot \right \|_{\infty} \leq C}(U) \right \|_F^2 + c^2 \left \| X\right \|_0 \right \}=T_c\left (\mathbb{P}_{\left \| \cdot \right \|_{\infty} \leq C}\left (U \right ) \right ),
\end{aligned}
\end{equation}
where the last equality is the definition of $T_c$ in Eq.~\eqref{hts}. 
\end{proof}

\subsection{The Corollary of the Proximal Mapping of $\ell_0$ Elastic Net Proposition and Its Proof}\label{sm2}
\begin{coro}\label{pro01}For given $U\in \mathbb{R}^{m\times n}$, the proximal mapping of $\Phi(S)=\alpha \left \| \bar{S}_0  \odot S\right \|_0 +  \beta \left \| S_0 \odot  S\right \|_0 + \gamma \left \| S\right \|_F^2$ on $\left \| S\right \|_{\infty}\leq C$  is 
\begin{equation}\label{proxs}
\begin{aligned}
\textup{prox}_{\left\|\cdot\right\|_{\infty}\leq C}(U,\alpha, \beta, \gamma)&\in \arg \min_{\left \| S \right \|_{\infty}\leq C}\left \{  \alpha \left \| \bar{S}_0  \odot S\right \|_0 +  \beta \left \| S_0 \odot  S\right \|_0 + \gamma \left \| S\right \|_F^2 + \left \| U-S \right \|_F^2\right \}\\
&=T_{\sqrt{\frac{\beta}{\gamma+1}}}\left (\mathbb{P}_{\left \| \cdot \right \|_{\infty}\leq C}\left ( \frac{\mathbb{U}}{\gamma+1}\right )\right )+T_{\sqrt{\frac{\alpha}{\gamma+1}}}\left (\mathbb{P}_{\left \| \cdot \right \|_{\infty}\leq C}\left (\bar{\frac{\mathbb{U}}{\gamma+1}}\right )\right ),
\end{aligned}
\end{equation}
where $\mathbb{U} = S_0\odot U$ and $\bar{\mathbb{U}} = \bar{S}_0\odot U$. 
\end{coro}

\begin{proof}
We know that the $U$ can be decomposed into $U=\textbf{1}\odot U = (S + \bar{S}_0)\odot U=S\odot U + \bar{S}_0 \odot U = \mathbb{U} + \bar{\mathbb{U}}$. Similarly, $S=S\odot S + \bar{S}_0 \odot S = \mathbb{S} + \bar{\mathbb{S}}$. Applying those decomposition into Eq.~\eqref{proxs}, we can decompose the proximal mapping into two parts.
\begin{equation}\label{prox1sm}
\begin{aligned}
&\arg \min_{ \left \| S \right \|_{\infty}\leq C}\left \{  \alpha\left \| \bar{S}_0  \odot S\right \|_0 +  \beta \left \| S_0 \odot  S\right \|_0 + \gamma \left \| S \right \|_F^2 + \left \| U-S \right \|_F^2\right \}\\
&= \arg \min_{\left \| \mathbb{S} \right \|_{\infty}\leq C}\left \{ \beta \left \| \mathbb{S}\right \|_0 + \gamma \left \| \mathbb{S} \right \|_F^2  +  \left \| \mathbb{U}-\mathbb{S} \right \|_F^2 \right \} + \arg \min_{\left \| \bar{\mathbb{S}} \right \|_{\infty}\leq C}\left \{  \alpha \left \| \bar{\mathbb{S}}\right \|_0 + \gamma \left \| \bar{\mathbb{S}} \right \|_F^2 + \left \| \bar{\mathbb{U}}-\bar{\mathbb{S}} \right \|_F^2 \right \} .
\end{aligned}
\end{equation}
Based on Proximal Mapping of $\ell_0$ Elastic Net Proposition~\ref{lemma2}, we know
\begin{equation}\label{prox2}
\begin{aligned}
\arg \min_{\left \| \mathbb{S} \right \|_{\infty}\leq a}\left \{ \beta \left \| \mathbb{S}\right \|_0 + \gamma \left \| \mathbb{S} \right \|_F^2  +  \left \| \mathbb{U}-\mathbb{S} \right \|_F^2 \right \} = T_{\sqrt{\frac{\beta}{\gamma+1}}}\left (\mathbb{P}_{\left \| \cdot \right \|_{\infty}\leq C}\left ( \frac{\mathbb{U}}{\gamma+1}\right )\right ).
\end{aligned}
\end{equation}
Similarly, 
\begin{equation}\label{prox3}
\begin{aligned}
\arg \min_{\left \| \bar{\mathbb{S}} \right \|_{\infty}\leq a}\left \{  \alpha \left \| \bar{\mathbb{S}}\right \|_0 + \gamma \left \| \bar{\mathbb{S}} \right \|_F^2 + \left \| \bar{\mathbb{U}}-\bar{\mathbb{S}} \right \|_F^2 \right \}  = T_{\sqrt{\frac{\alpha}{\gamma+1}}}\left (\mathbb{P}_{\left \| \cdot \right \|_{\infty}\leq C}\left (\bar{\frac{\mathbb{U}}{\gamma+1}}\right )\right ).
\end{aligned}
\end{equation}
Combining Eq.~\eqref{prox2} and Eq.~\eqref{prox3} proves the proposition. 
\begin{equation}\label{prox}
\begin{aligned}
& \arg \min_{\left \| S \right \|_{\infty}\leq C}\left \{  \alpha \left \| \bar{\mathbb{S}}_0  \odot S\right \|_0 +  \beta \left \| \mathbb{S}_0 \odot  S\right \|_0 + \gamma \left \| S\right \|_F^2 + \left \| U-S \right \|_F^2\right \}\\
                            &=T_{\sqrt{\frac{\beta}{\gamma+1}}}\left (\mathbb{P}_{\left \| \cdot \right \|_{\infty}\leq C}\left ( \frac{\mathbb{U}}{\gamma+1}\right )\right )+T_{\sqrt{\frac{\alpha}{\gamma+1}}}\left (\mathbb{P}_{\left \| \cdot \right \|_{\infty}\leq C}\left (\bar{\frac{\mathbb{U}}{\gamma+1}}\right )\right ).
\end{aligned}
\end{equation}

\end{proof}

\subsection{The convergence of \reNCA algorithm}\label{sm3} 
\begin{propname}[\textbf{Convergnece Proposition}]\label{converg}
Let $\left \{ (S^k, A^k)\right \}_{k\in\mathbb{N}}$ be a sequence generated by \reNCA algorithm. Then,\\
\textup{(i)} The sequence $\left \{ (S^k, A^k)\right \}_{k\in\mathbb{N}}$ has finite length, that is
\begin{equation}
\sum_{k=1}^{\infty} \left \| S^{k+1} - S^k \right \|_F + \left \| A^{k+1} - A^k \right \|_F < \infty.
\end{equation}
\textup{(ii)} The sequence $\left \{ (S^k, A^k)\right \}_{k\in\mathbb{N}}$ converges to a critical point $(S^*, A^*)$ of the \reNCA problem. 
\end{propname}
\begin{proof}
We apply Theorem 3.1 in~\cite{Bolte2014} to guarantee that the sequence generated by \reNCA is globally convergent to critical points of~\eqref{eqf}.
\end{proof}

\newpage
\begin{center}
\huge{\textbf{Supplementary Materials II:}}\\
\LARGE{\textbf{Implementation Details}}
\end{center}

\section{The NetREX Algorithm}\label{sm2c}

\subsection{The Details of the NetREX Algorithm}
\begin{algorithm}
    \SetKwInOut{Input}{Input}
    \SetKwInOut{Output}{Output}
    
    \Input{$S_0$, $E$, $L$, $\eta$, $\lambda$, $\kappa$, $\xi$, $v>0$ and $K$;}
    \Output{$S$ and $A$.}
    \Begin{
    $(S^0, A^0)=$Initialization$(S_0)$.\quad \quad \quad \quad \quad \quad \quad \quad \quad \quad \quad \quad \quad \quad \quad \quad \quad \ \ \ // Algorithm~\ref{alg2}. \\ 
    \For{$k=0, 1, 2, ..., K$}{
    $c^k = \max\left \{v, \ L(A^k)\right \}$. \quad \quad \quad \quad \quad \quad \quad \quad \quad \quad \quad \quad \quad \quad \quad \quad \quad \quad \\ 
    $U^k = S^k - \dfrac{1}{c^k}\left ( S^kA^k(A^k)^T + 2\kappa LS^k - E(A^k)^T \right )$. \quad \quad \quad \quad \quad \ \ \ // put~\eqref{dev1} into~\eqref{dev}.\\
    $S^{k+1}\in \textup{prox}_{\left\|\cdot\right\|_{\infty}\leq a}\left (U^k,\frac{2\eta}{c^k},  \frac{2(\eta-\lambda)}{c^k}, \frac{2\xi}{c^k}\right )$. \quad \quad \quad \quad \quad \quad \quad \quad \quad \quad \ \ //as shown in~\eqref{px1}.\\
    $d^k = \left \{v, \ L(S^{k+1})\right \}$. \quad \quad \quad \quad \quad \quad \quad \quad \quad \quad \quad \quad \quad \quad \quad \quad \quad \quad \ \ \ \\ 
    $V^k= A^k - \dfrac{1}{d^k}\left ((S^{k+1})^T(S^{k+1})A^k - (S^{k+1})^TE\right )$. \quad \quad \quad \quad \quad \quad \ \ // put~\eqref{dev1} into~\eqref{dev}.\\
    $A^{k+1} = \mathbb{P}_{\left \| \cdot \right \|_{\infty}\leq b}\left ( \dfrac{1}{1+\frac{2\mu}{d^k}}V^k \right )$. \quad \quad \quad \quad \quad \quad \quad \quad \quad \quad \quad \quad \quad \quad \ //as shown in~\eqref{px2}.
    }
    $S=S^K$ and $A=A^K$
    \caption{The \reNCA algorithm.}\label{alg1}
    }
\end{algorithm}

\subsection{The Initialization Algorithm for NetREX}\label{sm4}

To ensure that the starting point is consistent with the prior network,   $(S^0,A^0)$ have to  be inferred from our prior network $\mathcal{G}_0$. To do this, we compute $(S^0,A^0)$ by solving the following problem, which is obtained from the original \reNCA formulation  by dropping the constraints and disregarding the non-smooth regularization term $ \left (\eta-\lambda\right )\left \| S\odot S_0 \right \|_0+\left ( \eta+\lambda \right )\left \| S \odot \bar{S}_0 \right \|$ of $S$.
\begin{equation}\label{obj5}
\begin{aligned}
\min_{S, A}: & \ J(S,A) =\dfrac{1}{2}\left \| E - SA\right \|_F^2  +\kappa \textup{tr}(S^TLS) + \xi \left \|S\right \|_F^2 + \mu \left \| A\right \|_F^2
\end{aligned}
\end{equation}
The  problem~\eqref{obj5} can be solved by the standard Gauss-Seidel scheme~\cite{Grippo2000} that alternatively solves the multi-variable optimization problem with respect to one variable while fixing the rest of the variables. Specifically, we can fix $S=S^0_k$ and solve \eqref{obj5} with respect to $A$ in closed form shown in Line 4 of Algorithm~\ref{alg2}. Then, we fix $A=A_k^0$ and solve \eqref{obj5} with respect to $S$, whose solution is the solution of the Sylvester equation ${S}A_k^0(A_k^0)^T + 2(\kappa L + \xi I){S} = E(A_k^0)^T$ (derived by setting $\nabla H(S,A_k^0) = 0$). The Sylvester equation is solved by standard Bartels-Stewart algorithm. We alternatively ran lines 4 and 5 $K$ times. In the end, we project the solutions $A_K^0$ and $S_K^0$ into feasible space of Eq.~\eqref{eqf} by the projection operator~\eqref{proj0} shown in lines 7 and 8. Algorithm~\ref{alg2} elaborates the details of obtaining $(S^0,A^0)$. 

\begin{algorithm}
    \SetKwInOut{Func}{Function}
    \SetKwInOut{Input}{Input}
    \SetKwInOut{Output}{Output}

    \Func{Initialization$(S_0)$;}
    \Input{$S_0$;}
    \Output{$S^0$ and $A^0$.}
    \Begin{
    $S_0^0 = S_0$.\\
    \For{$k=0, 1, 2, ..., K$}{
    $A^0_k = \left ((S_k^0)^TS_k^0 + \mu I \right )^{-1} (S_k^0)^TE$.\\
    $S_{k+1}^0 := \left \{ \hat{S} | \hat{S}A_k^0(A_k^0)^T + 2(\kappa L + \xi I)\hat{S} = E(A_k^0)^T \right \}$.\\
    }
    $A^0 = \mathbb{P}_{\left \| \cdot \right \|_{\infty}\leq b}(A^0_K)$.\\
    $S^0 = \mathbb{P}_{\left \| \cdot \right \|_{\infty}\leq a} \left (S^0_K\right )$.
    \caption{The initialization for \reNCA}\label{alg2}
    }
\end{algorithm}

\section{Evaluation Metrics}

\subsection{F-measure}\label{smb1}
F-measure is defined as 
\begin{equation}
\textup{F-measure} = 2\times \dfrac{\textup{Precision} \times \textup{Recall}}{\textup{Precision} + \textup{Recall} },
\end{equation}
where 
\begin{equation}
\textup{Precision} = \dfrac{\left | \mathcal{E}^p \cap \mathcal{E} \right |}{\left | \mathcal{E}^p \right |}, \ \textup{Recall} = \dfrac{\left | \mathcal{E}^p \cap \mathcal{E} \right |}{\left | \mathcal{E} \right |}.
\end{equation}
$\mathcal{E}$ and $\mathcal{E}^p$ are edge sets of the underling regulatory network $\mathcal{G}$ and the predicted regulatory network, respectively. F-measure ranges from 0 to 1, where 1 presents that the underlining $\mathcal{G}$ is fully recovered and 0 means the opposite.

\subsection{Fold Enrichment for GO annotations}\label{smb2}
We consider two genes are co-regulated if the Jaccard similarity coefficient between the TF set regulating the first gene  and the TF set regulating the second gene is larger than 0.5. The Jaccard similarity coefficient between two sets is the ratio of the size of the intersection of the given two sets to the size of the union of these two sets. Then for each co-regulated gene pair, we again use the Jaccard similarity coefficient to measure the similarity between the GO annotation set corresponding to the first gene and the GO annotation set corresponding to the second gene. In the end, we compute the average of this coefficient overall co-regulated gene pairs. The same procedure was done for 100 randomized networks, and the enrichment is the ratio of the average coefficient of the original network to the average of the randomized networks. The randomized networks are generated by permuting the node labels of the original network. Hence, all randomized networks have the same topology to the original network but with different node labels. The statistical significance is accessed at a level of 0.05 using a one-side unpaired T-test for comparing the Jaccard coefficients from the original network with coefficients from 100 randomized networks. 

\subsection{Fold Enrichment for PPIs}\label{smb3}
Enrichment of co-regulated genes for PPIs was computed analogously to enrichment for GO annotations. Specifically, we computed the ratio of the number of PPIs for co-regulated gene pairs to the average number of such PPIs in 100 randomized networks, using the same definition for co-regulation and network randomization.

\section{Parameters}

\subsection{The \reNCAs\_NP and \reNCAs\_$\ell_1$ Algorithms}\label{smc1}
The \reNCAs\_NP algorithm is same to Algorithm~\ref{alg1} with $\lambda=0$. The formulation of \reNCAs\_NP is
\begin{equation}\label{eqfnp}
\begin{aligned}
\min_{S,A}\ &\dfrac{1}{2}\left \| E - SA \right \|_F^2+ \eta\left \| S \right \|_0  + \kappa \textup{tr}(S^TLS) + \xi \left \|S\right \|_F^2 + \mu \left \| A\right \|_F^2\\
       s.t.\ & \left \| S\right \|_{\infty} \leq a, \ \left \| A\right \|_{\infty} \leq b.
\end{aligned}
\end{equation}
The formulation is similar to sparse coding~\cite{Mairal2009} if we remove the graph embedding term. 

\noindent The \reNCAs\_$\ell_1$ formulation is as following.
\begin{equation}\label{eqfl1}
\begin{aligned}
\min_{S,A}\ &\dfrac{1}{2}\left \| E - SA \right \|_F^2+ \left (\eta-\lambda\right )\left \| S\odot S_0 \right \|_1+\left ( \eta+\lambda \right )\left \| S \odot \bar{S}_0 \right \|_1  + \kappa \textup{tr}(S^TLS) + \xi \left \|S\right \|_F^2 + \mu \left \| A\right \|_F^2\\
       s.t.\ & \left \| S\right \|_{\infty} \leq a, \ \left \| A\right \|_{\infty} \leq b.
\end{aligned}
\end{equation}
To do a fair comparison, we also use the PALM algorithm to solve it which is analogous to Algorithm~\ref{alg1}. The only difference is that in line 6 of Algorithm~\ref{alg1}, we use proximal mapping of  $\ell_1$ elastic net that is given in~\cite{Parikh2013} instead of proximal mapping for $\ell_0$ elastic net.  

\subsection{Parameter Settings for Simulated Data}
The parameters used to generate simulated data is $L=60, N=500, M=100$. The density of the ground truth GRN is 0.1. The noise level in simulated expression data $E$ is controlled by $\sigma = \{0,0.1,0.2,0.3,0.4,0.5,0.6,0.7,0.8,0.9,1.0\}$. The percentage of true edges in $\mathcal{G}_0$ is controlled by $\theta=\{0.1,0.2,0.3,0.4,0.5,0.6,0.7,0.8,0.9\}$.

There are seven parameters for \reNCA algorithm, which are $\lambda$, $\eta$, $\kappa$, $\xi$, $\mu$, $a$ and $b$. We applied grid search to find the optimal parameters. The settings are as following. We set $\eta -\lambda \in [0.2, 5]$ with interval 0.2, $\eta+\lambda \in [1,50]$ with interval 1, $\kappa \in [0.1, 0.5]$ with interval 0.1, $\xi=\{0.1, 1\}$, $\mu=\{0.1, 1\}$ and $a=b=\max_{i,j}(\textup{abs} \ E(i,j))$. We used the same parameter setting for \reNCAs\_NP except $\lambda=0$ and $\eta\in [1,50]$ with interval 1. For \reNCAs\_$\ell_1$ algorithm, we use exactly the same parameters to the \reNCA algorithm. 

To test the potential of the competing algorithms, for certain noise level, we first applied grid search for all algorithms to find their optimal parameters on only one simulated data set based on F-measure. Then we use the optimal parameters to other 50 simulated data set under the same noise level. We compared the performance of different algorithms based on the F-measures.

\subsection{Parameter Settings for DrosDel Data}\label{smc3}
We set $\eta-\lambda \in [0.01,0.2]$ with interval 0.01, $\eta+\lambda \in [0.5,10]$ with interval 0.5, $\kappa=\{0.05, 0.1\}$, $\xi=\{0.1, 1, 5\}$, $\mu=\{0.1,1\}$ and $a=b=\max_{i,j}(\textup{abs} \ E(i,j))$. Because we do not know the ground truth regulatory networks, we proposed a heuristic score to rank our predicted networks. The score can be computed as
\begin{equation}\label{hscore}
R=\left \| E - SA \right \|_F^4\times \left \| S \right \|_0.
\end{equation}
We reasoned that the promising networks should be able to describe the underling regulatory system (making $\left \| E - SA \right \|_F$ small) as well as have only the contributing regulations (the number of edges in the network $\left \| S \right \|_0$ is small). We used power of 4 on fitting error ($\left \| E - SA \right \|_F$) because, to build a condition specific RGN, fitting the condition specific expression data $E$ is more important. The smaller $R$ implies that we can fit the expression data using the network with smaller number of edges. We ranked all predicted networks under different parameters in terms of $R$ in ascending order. We showed the performance of the top 20 networks.        

For constructing the TF-gene correlation networks, we used Pearson coefficient cutoffs ($\{0.6, 0.7, 0.8, 0.9\}$) and show the performance of the networks under different cutoffs. 

For GENIE3, there is only a parameter $K$ used by it. \cite{Huynh-Thu2010} suggests two settings for $K$, which are $K=M-1$ and $K=\sqrt{M}$. We compared the results of these two $K$s. $K=M-1$ is better than $K=\sqrt{M}$. Therefore, we use $K=M-1$ in comparison. We also need a cutoff to get the final GRN. We ranked the weighted predicted by GENIE3 and used the top 100,000, 200,000, 300,000, 400,000 and 500,000 as output, respectively. 

The co-expressed gene pairs used in the comparison shown in Fig.~\ref{flydata} is the same to the one we inputted in our formulation as the graph embedding term. The Pearson coefficient cutoff used here is 0.88.  

\subsection{The Best Network Based On Fold Enrichment}\label{smc4}
For the correlation based method, GENIE3 and \reNCA, we show their performance under different parameters in fold enrichment analysis shown Fig.~\ref{flydata}. When comparing their performance on agreement with DSX targets, we only use the networks with the best fold enrichment. We select the best networks as following. Take the networks predicted by \reNCA for example. First, we ranked all networks based on fold enrichment of GO annotations and stored in $R_{GO}$. Then we ranked all networks based on fold enrichment of PPIs and stored in $R_{PPIs}$. We then ranked the networks based on the sum of the ranking we just computed ($R_{GO}+R_{PPIs}$) and treat the top network as the best network.

\end{document}